\documentclass[11pt]{article}
\usepackage{packages}
\usepackage{statistics-macros}
\usepackage{editing-macros}
\usepackage{formatting}

\title{Explaining Practical Differences Between Treatment Effect Estimators with High Dimensional Asymptotics}
\author{Steve Yadlowsky\\\texttt{yadlowsky@google.com}\\Google Research, Brain Team}
\date{\today}

\begin{document}

\maketitle

\begin{abstract}
    We revisit the classical causal inference problem of estimating the average treatment effect in the presence of fully observed confounding variables using two-stage semiparametric methods. In existing theoretical studies of methods such as G-computation, inverse propensity weighting (IPW), and two common doubly robust estimators---augmented IPW (AIPW) and targeted maximum likelihood estimation (TMLE)---they are either bias-dominated, or have similar asymptotic statistical properties. However, when applied to real datasets, they often appear to have notably different variance. We compare these methods when using a machine learning (ML) model to estimate the nuisance parameters of the semiparametric model, and highlight some of the important differences. When the outcome model estimates have little bias, which is common among some key ML models, G-computation and the TMLE outperforms the other estimators in both bias and variance.
    
    We show that the differences can be explained using high-dimensional statistical theory, where the number of confounders $d$ is of the same order as the sample size $n$. To make this theoretical problem tractable, we posit a generalized linear model for the effect of the confounders on the treatment assignment and outcomes. Despite making parametric assumptions, this setting is a useful surrogate for some machine learning methods used to adjust for confounding in two-stage semiparametric methods. In particular, the estimation of the first stage adds variance that does not vanish, forcing us to confront terms in the asymptotic expansion that normally are brushed aside as finite sample defects. However, our model emphasizes differences in performance between these estimators beyond first-order asymptotics.
\end{abstract}

\section{Introduction}
We investigate and compare the statistical properties of two-stage estimators of the average treatment effect (ATE), $\tau = \E[Y(1) - Y(0)]$, of a binary treatment $W \in \{0, 1\}$ on a scalar outcome $Y \in \R$, in the presence of fully observed, confounding variables $X \in \R^d$ that may affect both treatment assignment and the outcome. For this class of treatment effect estimation problems, we compare a number of commonly used estimators for treatment effects, including G-computation \citep{Robins86}, inverse probability weighting, and two common doubly robust estimators---augmented inverse probability weighting \citep{bang2005doublerobust} and targeted maximum likelihood estimation \citep{van2006targeted}. Theoretical analyses have shown that when each of these is paired with an appropriate nonparametric or machine learning estimator of the relevant nuisance parameters in the semiparametric model, it can achieve the same (optimal) semiparametric efficiency bound \cite{Imbens04,van2006targeted,ChernozhukovChDeDuHaNeRo16}. That is, in the large sample limit, they all have equivalent statistical properties. Additionally, analyses show that the doubly robust estimators share similar statistical properties in sparse, high dimensional linear models for the covariates, when the LASSO or related sparsity-inducing machine learning methods are used for nuisance parameter estimation \cite{BelloniChHa14}, but G-computation is dominated by bias due to heavy regularization of the nuisance parameter models.

Yet, empirical evaluations of such methods suggest that this theoretical similarity does not translate to practice.
Take, for instance, the results of the Atlantic Causal Inference Conference's Data Analysis Challenge from 2019. Numerous semiparametric methods were implemented in a variety of forms, including G-computation, TMLE, and AIPW. While all of these methods have fairly low mean-squared error (MSE), they do not all have exactly the same performance. While the semiparametric methods are modular, so that all of them can be used with the same models for nuisance parameter estimation, these were not directly compared for most methods in the ACIC challenge. We begin our comparison by developing a realistic simulation model on which to evaluate these methods in a modular way, and identify some concrete differences not predicted by the aforementioned asymptotic theory.

This suggests that theoreticians need to study these methods under different modeling assumptions for the data, that correctly predicts these practical differences. Here, we propose that assuming that the observed confounders are in a high dimensional regime is a better model. By high dimensions, we mean that the number of confounders $d$ is of the same order as the sample size $n$, so that $d/n \to \kappa \in (0, 1)$. To make the problem tractable, we posit the (strong) assumption that all nuisance parameters, specifying the effect of the confounders on the treatment assignment and outcomes, follow an appropriate generalized linear model,
\begin{align}
    &\E[W \mid X=x] = h^{-1}(\eta^\top x), \label{eq:prop-structural}\\
    &\E[Y(1) \mid X=x] = g^{-1}(\beta_1^\top x),~\text{and} \label{eq:treated-structural}\\
    &\E[Y(0) \mid X=x] = g^{-1}(\beta_0^\top x). \label{eq:control-structural}
\end{align}
In contrast to existing literature on average treatment effect estimation in high dimensions \cite{AtheyImWa18,WangHeXu20,ChernozhukovChDeDuHaNeRo16}, we do not assume that the effect of the confounders on the outcome or treatment is sparse. Instead, we only require that the magnitude of the nuisance parameters remains non-degenerate. That is, we study sequences such that $\eta^\top X = \Theta_P(1)$ as $n,d \to \infty$, and similarly for $\beta_1$ and $\beta_0$. Doing so makes the problem tractable, avoiding the issue of vanishing overlap between treated and control individuals in covariate space raised by \citet{DAmour20}. 

Despite making parametric assumptions, the high dimensionality makes this setting have some important similarities with using machine learning methods to adjust for confounding. In particular, the estimation of the nuisance parameters involves over-fitting the observed data, which introduces excess variance in the nuisance parameter estimates, modulated by the aspect ratio $\kappa$, that does not vanish asymptotically. In essence, this forces us to confront terms related to the variance of the nuisance parameter estimates that normally are brushed aside as finite sample defects. In many settings, ML estimators work best with hyperparameters tuned to have small bias, at the cost of higher variance. Choosing a parametrically well-specified model with high variance reflects these qualitative properties while enabling tractable mathematical analysis. This perspective on asymptotics is used in the machine learning literature to understand important phenomena about machine learning methods poorly explained by traditional asymptotic or finite sample analyses \citep{MontenariRuSoYa19,WangMuTh21,MontenariZh20,TriperaneniAdPe21}.

In Section~\ref{sec:linear}, we perform a theoretical analysis of these estimators using a linear outcome model (i.e. $h$ and $g$ are both the identity link function) to illustrate some of the advantages and disadvantages of each estimator that stand out in this high-dimensional setting. In Section~\ref{sec:logistic}, we discuss extensions to a logistic regression outcome model. In this model, overfitting during nuisance parameter estimation due to the high dimensional covariates induces variance \emph{and} bias in the nuisance parameter estimates that must be corrected. We show that recalibrating the logistic regression estimates to have good calibration does \emph{not} properly correct for this bias, but that using the rescaling proposed by \citet{SurCa2019} and implemented using SLOE \citep{YadlowskyYuMcDA21} does correct for this bias. 

In Sections~\ref{sec:linear-experiments} and \ref{sec:logistic-experiments} on linear and logistic regression models, respectively, we compare many of the estimators discussed to one another via simulation, and compare them to oracle estimators with known nuisance parameters to emphasize the differences created by high dimensional nuisance parameter estimation. %

In comparing two-stage estimators for the ATE in this high dimensional setting, we find that when the outcome model estimates are unbiased, the G-computation estimator works very well, outperforming all of the other estimators in terms of bias and variance. However, when using biased estimates (such as those adjusted by post-hoc Platt scaling) of the nuisance parameters, the bias of G-computation can be much larger than doubly-robust estimators. Among the doubly robust estimators, the TMLE estimator has the lowest variance. These results match the phenomena uncovered in our simulation of methods when using ML models for nuisance parameter estimation.

\subsection{Related Work}
This work draws on, and is closely related to, the work on estimation and inference of the parameters of high-dimensional linear \cite{ElKarouiBeBiLiYu2013, DonohoMo2016, DobribanWa2018, CattaneoJaNe18a}, logistic \cite{SurCa2019}, and partially linear \cite{CattaneoJaNe18b} regression. It is different from these works in that we do not emphasize estimation or inference of the parameters of the parametric model, but rather the average treatment effect functional. In particular, this allows the generality of interaction terms between the treatment and confounders in linear models, and allows estimation in logistic regression models of a parameter that can be interpreted directly, without issues of baselining effects \cite{robins1997toward}. In the context of the extensive literature on large sample theory for ATE estimation, this allows us to compare ATE estimators from a fresh perspective, under which we will see that many estimators with similar large sample properties, including the G-computation, TMLE, and AIPW estimators, have different properties from one another in the high-dimensional asymptotic regime.

This work is closely related to the work by \citet{CattaneoJa18}, which studies semi-parametric inference with undersmoothed kernel estimators of the nuisance parameters. As noted in that paper, as well as  Enno Mammen's discussion of \citet{CattaneoCrJa13}, undersmoothed kernel estimators share some similarities to robust inference in high dimensions. However, in these related works, the conditions ensure consistency of the nuisance parameter estimates. This would correspond in high dimensions to assuming that $d/n = o(1)$, albeit possibly slowly. Other work such as \citet{RobinsTcLiVa09} study semiparametric functionals such as the ATE when the nuisance parameters are not smooth enough for a nonparametric estimator to achieve fast enough rates of convergence. This leads to slower than $\sqrt{n}$ rates of convergence for the ATE estimate. The high-dimensional setting we study challenges standard asymptotic results differently, because we demonstrate estimators achieving the $\sqrt{n}$ rate of convergence, but these estimators have higher variance than predicted by large sample asymptotic theory. \citet{CattaneoJaMa19} study this high dimensional regime, providing interesting methods for estimation and inference; building on this work, we emphasize how this regime can be used to compare and contrast a variety of two-stage semiparametric methods for treatment effect estimation.

\section{Preliminaries}
\label{sec:prelim}
We begin by giving a more detailed description of each of the semiparametric methods that we study in this work, and review some of the existing asymptotic theory for these estimators.

\subsection{G-computation}
Consider the following procedure for estimating the ATE. Apply data-splitting by randomly splitting the data into two pieces, with $\mathcal{I}_1$ the set of indices in the first split and $\mathcal{I}_2$ the indices in the second. Use the data in $\mathcal{I}_1$ to fit $\what{\mu}_w(x)$ estimating $\mu_w(x) = \E[Y \mid W=1, X=x]$. Similarly, fit $\what{\mu}_0(x)$ estimating $\E[Y \mid W=0, X=x]$. Then, use the G-computation estimator to estimate the average treatment effect on the individuals in $\mathcal{I}_2$, by imputing the potential outcomes:
\begin{equation}
    \what{\theta} = \frac{1}{n_2} \sum_{i \in \mathcal{I}_2} \what{\mu}_1(X_i) - \what{\mu}_0(X_i).
    \label{eq:imputation}
\end{equation}

This method for estimating the average treatment effect is anything but new \cite{Robins86}. However, existing literature focuses on models where the nuisance parameter estimates converge, asymptotically, to the true nuisance parameters. That is, some condition like $\|\what{\mu}_w(\cdot) - \mu_w(\cdot)\|^2 = \E[ ((\what{\beta}_w - \beta_w)^\top X)^2 \mid \what{\beta}_w ] = o_P(1)$ (or perhaps a stronger rate). However, as we will show in Section~\ref{sec:linear-g-est-high-theory}, this will not happen in our high dimensional regime.

\subsection{IPW Estimator}
Applying the same general technique as above to fit $\what{\pi}(x)$ for estimating the the propensity score $\pi(x) = P(W=1 \mid X=x)$, and using inverse propensity weighting (IPW), we can estimate the treatment effect as
\begin{equation*}
    \frac{1}{n_2} \sum_{i \in \mathcal{I}_2} \frac{W_i Y_i}{\what{\pi}(X_i)} - \frac{(1-W_i)Y_i}{1-\what{\pi}(X_i)}.
\end{equation*}
While theoretical analyses show that this method can estimate the treatment effect efficiently \citep{HiranoImRe03}, it suffers from more theoretical and practical issues than the others discussed. First of all, in both the large sample and high-dimensional asymptotic regimes, the IPW estimator has higher variance than G-computation, as it does not remove reducible variance in the outcomes. Therefore, the IPW estimator is typically inefficient. In the large sample asymptotic regime, nonparametric modeling of the propensity score can eliminate this excess variance \citep{HiranoImRe03}, although such nonparametric modeling is not feasible in high dimensions.

This method has additional estimation issues in the high dimensional regime studied here. As we will lay out in Section~\ref{sec:logistic}, the maximum likelihood estimate (MLE) of $\what{\eta}$ is biased when $d \asymp n$ (see \citet{SurCa2019} for a detailed treatment of this setting). Second of all, because the propensity score sits in the denominator of the IPW estimating equation, variance in the estimated propensity score will create bias in the estimating equation. Because the variance does not disappear asymptotically, this bias will remain non-negligible for large $n$. In our more detailed discussion of estimating nuisance parameters that follow a logistic regression model in Section~\ref{sec:logistic}, we show how to correct for both of these biases under the logistic regression model for the propensity score.

\subsection{Augmented IPW Estimator}
\label{sec:linear-aipw}
The celebrated doubly-robust augmented IPW estimator (AIPW) combines ideas from G-computation and the IPW estimator,
\begin{equation}
    \what{\theta}_{\aipw{}} = \frac{1}{n} \sum_{i=1}^n \what{\mu}_1(X_i) - \what{\mu}_0(X_i) + \frac{W_i - \what{\pi}(X_i)}{\what{\pi}(X_i)(1 - \what{\pi}(X_i))}(Y_i - \what{\mu}_{W_i}(X_i).
    \label{eq:aipw}
\end{equation}
Using sample splitting to fit the nuisance parameters for this estimator, allows us to use the product of biases of the nuisance parameters to bound the estimator's bias. Let $\overline{\mu}_w(x) = \E[\what{\mu}_w(x)]$ and $\overline{\pi^{-1}}(x) = \E[\frac{1}{\what{\pi}(x)}]$. Then,
\begin{align}
    |\E[\what{\theta}_{\aipw{}}] - \theta| &= |\E[ (\overline{\mu}_w(X) - \mu_w(X))(1 - \pi(X)\overline{\pi^{-1}}(X))]| \\
    &\le \E[ (\overline{\mu}_w(X) - \mu_w(X))^2 ] \E[(1 - \pi(X)\overline{\pi^{-1}}(X))^2].
\end{align}
This has two well-known implications. Firstly, if a well-specified parametric model is used to fit either the outcome imputation model or the propensity score model, then the resulting AIPW estimator will be asymptotically unbiased, i.e. robust in spite of the potential mispecification of the other nuisance parameter; this is why the estimator is called doubly-robust. Secondly, if $\max_w \|\what{\mu}_w(\cdot) - \mu_w(\cdot)\|_{2, P}\|\what{\pi}(\cdot) - \pi(\cdot)\|_{2, P} = o_P(n^{-1/2})$ and $\what{\pi}(\cdot)$ is uniformly bounded away from 0 or 1, then the estimator will be asymptotically unbiased, which justifies the use of many nonparametric or machine learning estimators of the nuisance parameters under the appropriate structural assumptions in conjunction with the AIPW estimator \cite{ChernozhukovChDeDuHaNeRo16}.

When $\mu_w(\cdot)$ and $\pi(\cdot)$ are modeled nonparametrically (or with a well-specified parametric model), the estimator $\what{\theta}_{\aipw{}}$ will be approximately (or exactly) unbiased, and so the key to understanding it's convergence to $\theta$ is to understand it's variance. When $d$ is constant, $\what{\theta}_{\aipw{}}$ is asymptotically equivalent to the oracle AIPW,
\begin{equation}
    \what{\theta}_{\oracle} = \frac{1}{n} \sum_{i=1}^n \mu_1(X_i) - \mu_0(X_i) + \frac{W_i - {\pi}(X_i)}{{\pi}(X_i)(1 - \pi(X_i))}(Y_i - \mu_{W_i}(X_i)),
    \label{eq:aipw-oracle-linear}
\end{equation}
which has variance
\begin{equation*}
    \frac{1}{n}\left( \var((\mu_1(X) - \mu_0(X)) + \E\left[ \frac{\var(Y(1) \mid X)}{e(X)} + \frac{\var(Y(0) \mid X)}{1-e(X)} \right]\right),
\end{equation*}
and achieves the semiparametric efficiency bound \cite{hahn98}. Under sufficient smoothness and regularity conditions, the G-computation and AIPW estimators with nonparametric nuisance parameter estimates will have nearly identical large sample behavior. However, when $\mu_w(\cdot)$ is well specified by a linear model, G-computation will generally have lower variance than the AIPW estimator \citep{AtheyImWa18}, except under special conditions on the propensity score such as in a completely randomized experiment ($\pi(x) = 0.5$).

\subsection{Targeted Maximum Likelihood Estimation}
\label{sec:tmle-linear}
The targeted maximum likelihood estimator (TMLE), like the AIPW, is a doubly robust estimator of the ATE. Under many settings, theoretical analyses show that the TMLE and AIPW estimators are asymptotically equivalent, sharing the same asymptotic influence function, however empirical findings suggest differences between the two estimators in finite sample performance.

Let us revisit the TMLE of the ATE with continuous outcomes in our notation, using a sample split variant to simplify the analysis. Using the first split, $\mathcal{I}_1$, estimate the nuisance parameters $\what{\beta}_0, \what{\beta}_1,$ and $\what{\eta}$. Let $\what{\mu}_1(x), \what{\mu}_0(x),$ and $\what{\pi}(x)$ be the estimated nuisance parameter predictions. Then, on second split, $\mathcal{I}_2$, fit $\what{\epsilon}_0$ and $\what{\epsilon}_1$ as the MMSE regression estimates of the fluctuations $\frac{W}{\what{\pi}(X)}$ and $\frac{1-W}{1-\what{\pi}(X)}$, respectively, on the outcome residuals $Y - \what{\mu}_W(X)$,
\begin{equation*}
    \what{\epsilon}_0, \what{\epsilon}_1 = \argmin_{\epsilon_0,\epsilon_1} \frac{1}{n_2} \sum_{i \in \mathcal{I}_2} \left(Y_i - \what{\mu}_{W_i}(X_i) + \epsilon_1 \frac{W_i}{\what{\pi}(X_i)} + \epsilon_0 \frac{1-W_i}{1-\what{\pi}(X_i)}\right)^2.
\end{equation*}
Finally, using the updated regression estimates $\overline{\mu}_w(x) = \what{\mu}_w(x) + \what{\epsilon}_1 \frac{w}{\what{\pi}(x)} + \what{\epsilon}_0 \frac{1-w}{1-\what{\pi}(x)} $, estimate the TMLE as
\begin{equation}
    \what{\tau}_{\tmle{}} = \frac{1}{n_2} \sum_{i \in \mathcal{I}_2}  \overline{\mu}_1(X_i) - \overline{\mu}_0(X_i).
\end{equation}

\section{Comparing Methods using ML}
\label{sec:ml-sim}
Here, we perform a realistic simulation study comparing the semiparametric methods discussed in Section~\ref{sec:prelim}, and conclude that there are notable differences in their performance in practice, which are not explained by existing large-sample theory.
In many real-world applications of these methods, the nuisance parameters need to be flexibly estimated, and the accuracy of estimating these nuisance parameters is paramount ensuring accurate treatment effect estimates \citep{DorieHiShScCe19}. Recent advances in machine learning technology provide a wide range of methods that can be applied for this purpose. An advantage of the semiparametric methods discussed in Section~\ref{sec:prelim} over other bespoke methods is that they are modular, so that the best nuisance parameter estimator for the data can be used with any of the methods. Unfortunately, the results on standard benchmarks in causal inference, such as the ACIC Data Challenges, only compare some of these methods side-by-side.

In this section, we directly compare the semiparametric methods on simulated data using the same ML nuisance parameter estimates. To make the results of this comparison as relevant to practice as possible, we chose a popular machine learning method that performs well on a wide variety of datasets. Then, we chose one such dataset as the basis for a semi-synthetic simulation, designed to ensure that the nuisance parameter estimation error is realistic.

We generated 500 datasets using this simulation, estimated the nuisance parameters using cross fitting on each dataset, and estimated the treatment effect using all of the semiparametric methods described in Section~\ref{sec:prelim} with these shared nuisance parameter estimates. We then summarized the bias, standard error, and root-MSE (RMSE) of each, reported below in Section~\ref{sec:ml-sim-results}.

\subsection{Machine Learning Method Selection}
\label{sec:ml-sim-ml}
Gradient boosted decision trees and random forests both provide good prediction performance on a wide variety of publicly available machine learning datasets \cite{ShwartzZivAr22,FernandezCeBaAm14}. In this work, we use \texttt{xgboost} \citep{ChenGu16} for all nuisance parameter estimates, as it has good prediction accuracy and is computationally efficient \citep{ShwartzZivAr22}. \texttt{xgboost} has many hyperparameters that need to be tuned to get good performance. For our experiments, we fit nuisance parameter estimates on all simulations with a variety of randomly sampled hyperparameter combinations, and found the choice that optimizes the average MSE of the treatment effect estimates for the G-estimation, AIPW, and TMLE methods. This corresponds to an (approximate) oracle selection of hyperparameters for the overall treatment effect estimation problem. Additionally, we performed cross-validation while estimating the nuisance parameters, and chose the hyperparameters that minimize the prediction loss for the corresponding learning task. We report results in Section~\ref{sec:ml-sim-results} for both approaches.

\subsection{Simulation Design}
\label{sec:ml-sim-design}
In choosing a simulation to perform the comparison of treatment effect estimation methods, a few properties are important to make the experiment relevant to practice, which we outline below. We are unaware of any established and publicly available simulations with these properties, so we opted to create our own. As noted in \citet{DorieHiShScCe19}, a critical feature of practical treatment effect estimation methods is the use of a flexible and accurate outcome model. The G-computation, AIPW and TMLE methods all use outcome models, albeit in different ways. Therefore, an important aspect of these estimators to compare is how they respond to realistic outcome models and the corresponding estimation errors. Therefore, an important property of the simulations is the realism of the relationship between the outcome and observed confounding variables. Beyond this, the simulation must have a realistic level of confounding; if the treatment is unrelated to the confounders, then adjustment for confounding using any of these methods is unnecessary.

To create a realistic data simulation, we used a semi-synthetic design, where some of the variables are observed from real data and others are simulated. We use the Wine Quality data \cite{CortezCeAlMaRe09} from the UCI Machine Learning Repository as the basis for the observed data. This dataset has a variety of tabular, measured properties of vinho verde wines from Portugal, along with ordinal assessments of the wines' quality on a scale of 1 to 10; see \citet{CortezCeAlMaRe09} for more information. We imagine a causal inference problem where some (hypothetical) wine preservation technique is used on some of the wines, after they are produced, but before their quality is assessed. In this hypothetical setting, we are interested in estimating the causal effect of the wine preservation technique on the wine quality. Confounding comes into play, because the wines that are preserved may be selected based on their perceived future quality.

To implement this simulation setting, for each white wine in the dataset, we sample the hypothetical treatment on the basis of a propensity score model learned from an auxiliary task described below. For all treated units, we increase the wine quality score by $1$, so that the treatment effect is known to be exactly $1$.

The auxiliary task for generating the propensity score that we used was to build a predictive model on the red wine data to predict whether the wine had a quality score above or below the median score in the data. The predicted probabilities on the white wine data according to this model were re-scaled so that the log odds of treatment had a variance of $1$, and this was used as the propensity score in simulation.

This approach ensures that the joint distribution of $Y(0)$ and $X$ are as realistic as possible, coming directly from observed data. The treatment is generated synthetically, it is a realistic function learned from a flexible estimation procedure on a statistically independent dataset.

\subsection{Results}
\label{sec:ml-sim-results}
The results of this simulation study, shown in Table~\ref{tab:ml-sim-compare}, indicate practical differences among the semiparametric methods considered here. Overall, we see that the estimation error is sensitive to the hyperparameter tuning methods, and by extension, the nuisance parameter estimation. However, some methods are more sensitive to these issues than others. For example, the TMLE and G-computation estimators outperform the AIPW and IPW estimators in most properties. The very poor performance of the AIPW and IPW estimators in Table~\ref{tab:ml-sim-compare}(a) is likely due to propensity score estimation errors that have a high leverage on these methods.

\begin{table}[t]
\caption{The method performance in simulations doesn't match the predictions made by existing asymptotic theory. The simulation design is described in Section~\ref{sec:ml-sim-design} with ML estimation of nuisance parameters, described in Section~\ref{sec:ml-sim-ml}.}
\begin{subtable}[t]{0.45\textwidth}
\begin{tabular}[t]{lrrr}
\toprule
Method &      Bias &       \makecell[b]{Standard\\Error} &      RMSE \\
\midrule
G-computation   & -0.0194 &  0.0337 &  0.0389 \\
TMLE          & -0.0061 &  0.0388 &  0.0392 \\
AIPW          & -0.0018 &  0.0437 &  0.0437 \\
IPW           &  0.1124 &  0.2088 &  0.2371 \\
\bottomrule
\end{tabular}
    \caption{Results when ML hyperparameters are selected to be the best for treatment effect estimation. TMLE and G-computation perform well in both bias and variance.}
\end{subtable}
\hspace{\fill}
\begin{subtable}[t]{0.45\textwidth}
\begin{tabular}[t]{lrrr}
\toprule
Method &      Bias &       \makecell[b]{Standard\\Error} &      RMSE \\
\midrule
G-computation   & -0.0415 &  0.0329 &  0.0530 \\
TMLE         & -0.0130 &  0.0381 &  0.0403 \\
AIPW          & -0.0097 &  0.0400 &  0.0412 \\
IPW           & -0.1879 &  0.2330 &  0.2993 \\
\bottomrule
\end{tabular}
    \caption{Results when ML hyperparameters are selected to be the best for prediction performance in cross-validation. TMLE and AIPW perform best in RMSE, but  G-computation has the smallest standard error.}
\end{subtable}
    \label{tab:ml-sim-compare}
\end{table}

With a fixed choice of hyperparameters, the AIPW method has a 30\% larger standard error than the G-computation method; the TMLE is somewhere in-between these two. Even when using hyperparameters adaptively selected by cross-validation, where the AIPW method outperforms G-computation, the standard error of the AIPW method is 28\% larger than G-computation, and 8\% larger than the TMLE. The improved performance comes from it's lower bias. However, we note that the RMSE of the AIPW with these hyperparameters is still higher than G-computation or the TMLE with an appropriately chosen hyperparameter for the treatment effect estimation problem.

The bias reduction properties of doubly robust estimators have been studied extensively by many in the causal inference literature, see, for example, the discussion and references in \citet{ChernozhukovChDeDuHaNeRo16}. Here, we focus on some of the features of these methods that are not well-predicted by this theory. In particular, we emphasize the differences in variance. These are often de-emphasized in the literature, because in large-sample asymptotic regimes or high dimensional regimes with sparse nuisance parameter models, they are lower order terms that become negligible. However, in the high dimensional regime that we study in Sections~\ref{sec:linear} and~\ref{sec:logistic}, they are not negligible, making them easier to compare in theoretical analysis.

\section{Linear outcome model}
\label{sec:linear}
Now, we study and compare the statistical behavior of the aforementioned methods in the proposed high dimensional regime, and present our main results on the methods' performance in this setting. We begin with a simple model--linear regression, as a model for the potential outcomes, given the confounding variables. That is, we assume
\begin{align}
    \begin{cases}
    X \sim \normal{}(0, \Sigma) \\
    Y(w) = \beta_w^\top X + \epsilon,
    \end{cases}
    \label{eq:linear-model}
\end{align}
for $w = 0, 1$ and $\epsilon \sim (0, \sigma^2)$ drawn i.i.d., independent of $X$, with $\|\Sigma\|_{\mathrm{op}} \le M < \infty$.
This model serves as a proof of concept, to demonstrate that there are non-vacuous convergence results for estimating the ATE, even when estimates of the nuisance parameters do not converge to the true nuisance parameter, such as when $d$ is of the same order of magnitude as $n$. The key property provided by linear models is that overfitting the data when fitting nuisance parameters introduces excess variance, but not bias in the nuisance parameter estimates. The choice of the distribution of $X_i$ also makes analysis of many of the estimators straightforward; most of the results would also hold for $X_i$ a sub-Gaussian vector with an appropriately bounded Orlicz norm, and appropriate modification of the constants.

Because the treatment is binary, it is not generally appropriate to assume a linear model is parametrically well specified for the propensity score $\pi(x) = P(W = 1 \mid X = X)$. Therefore, throughout we will assume that the propensity score follows a logistic regression model. Following the discussion in \citet{SurCa2019}, we will assume that $\var(\eta^\top X) = \gamma^2 \in (0, \infty),$ for all $n$, so that the propensity score does not become degenerate, leading to $P(W=1 \mid X=x) \in \{0,1\}.$ With this choice of propensity score, the $\chi^2$-divergence between the covariate distributions $P_{X\mid W=w}(\cdot)$ of the treated ($w=1$) and control ($w=0$) populations is finite, even though strong overlap $0 < \underline{p} \le P(W=1 \mid X=x) \le 1-\underline{p}$ does not hold almost surely for any $p$; aside from the high-dimensionality, all of the necessary assumptions for nonparametric identification of the treatment effect hold in our model.

When assuming a linear model, $\what{\mu}_w(x)$ can be estimated by $\what{\beta}_w^\top x$, where $\what{\beta}_w$ is estimated by applying linear regression of $Y$ on $X$ among the individuals with $W=w$,
\begin{equation}
    \what{\beta}_w = \left(\sum_{i \in \mathcal{I}_1} \ind{W_i=w} X_i X_i^\top \right)^{-1}\sum_{i \in \mathcal{I}_1} \ind{W_i=w} X_i Y_i = \left( \what\Sigma_1 \right)^{-1} \sum_{i \in \mathcal{I}_1} \ind{W_i=w} X_i Y_i / N_{1w}.
    \label{eq:linear-outcome-mmse}
\end{equation}
with  $N_{1w} = \sum_{i=1}^n \ind{W_i=w}$ and $\what{\Sigma}_{w} = \frac{1}{N_{1w}}\sum_{i=1}^n \ind{W_i=w} X_i X_i^\top$.

\subsection{G-computation}
\label{sec:linear-g-est-high-theory}
Plugging in the above estimator \eqref{eq:linear-outcome-mmse} to the general G-computation estimator~\eqref{eq:imputation}, we get a simple expression for G-computation in the linear model, 
\begin{equation}
    \what{\theta} = \frac{1}{n_2} \sum_{i \in \mathcal{I}_2} \what{\beta}_1^\top X_i - \what{\beta}_0^\top X_i.
    \label{eq:imputation-linear}
\end{equation}
Existing theoretical analyses of G-computation focus on models where the nuisance parameter estimates converge according to some condition like $\|\what{\mu}_w(\cdot) - \mu_w(\cdot)\|^2 = \E[ ((\what{\beta}_w - \beta_w)^\top X)^2 \mid \what{\beta}_w ] = o_P(1)$ (or perhaps a stronger rate). The following proposition shows that this does not happen in our high dimensional setting.

\begin{proposition}
Assume that the data follow the model in \eqref{eq:linear-model} and for each $w \in \{0, 1\}$, let $\what{\beta}_w$ be as in \eqref{eq:linear-outcome-mmse}, with $d(n) / n \to \kappa \in (0, 1)$. Let $G$ be the event that $N_{1w} > d,$ so that $\what{\beta}_w$ is well-defined. Let $x_n \in R^{d(n)}$ be a sequence of test points such that $\| x_n \| \asymp \sqrt{d}$. Then,
\begin{equation}
    \liminf_{n \to \infty} \var(x_n^\top \what{\beta}_w \mid G) > 0.
\end{equation}
\label{prop:prediction-variance}
\end{proposition}
\noindent Proof of this result can be found in Section~\ref{sec:prediction-variance-proof}.

While the nuisance parameter estimates are inconsistent, they remain unbiased with bounded variance. This fact is key to giving our main result for this section. In the following theorem, we show that plugging these coefficients in to estimate the ATE leads to an asymptotically $\sqrt{n}$-consistent estimate of the ATE under regularity conditions preventing the least squares estimators $\what{\beta}_1$ and $\what{\beta}_2$ from being degenerate. The nature of these conditions is not a critical part of the technical development in this paper, and we abstract them into the following assumption:
\begin{assumption}
\label{assume:least-squares-good}
There exists $C_G, \sigma_{\max}^2 < \infty$ such that $P( \lambda_{\max}(\what\Sigma_w^{-1}) < C_G,~\text{for w=0,1}) \to 1$, and $\var(Y(w) \mid X=x) \le \sigma_{\max}^2$.
\end{assumption}
This assumption is satisfied when the entries of $X_i$ are sufficiently independent and have thin tails in both the treated and control samples, see~\citet{Vershynin12,Vershynin16} for details of the relevant random matrix theory.

\begin{theorem}
\label{thm:linear-moments}
Let $\what{\theta}$ be as defined in \eqref{eq:imputation-linear}, and let $\theta = \E[Y(1) - Y(0)]$. Assume the data follows the structural model in Eqs. \eqref{eq:prop-structural}-\eqref{eq:control-structural}, with $g^{-1}(t) = t$, $\var(X) = \Sigma$, having a bounded condition number.
Then, there exists a sequence of events $G_n$ measurable with respect to $(X_i, W_i)_{i \in \mathcal{I}_1}$, such that $P(G_n) \to 0$, $\E[\what{\theta} \mid G_n] = \theta$, and
\begin{align*}
    \var\left( \what{\theta} \mid G_n \right) &= (\beta_1 - \beta_0)^\top \Sigma (\beta_1 - \beta_0) / n_2 + \trace{}\left( \var\left(\what\beta_1 - \what\beta_0 \mid G_n \right)\Sigma \right) / n_2 \\
    &~\quad~\quad + \E[X]^\top \var(\what\beta_1-\what{\beta}_0 \mid G_n) \E[X].
\end{align*}
If Assumption~\ref{assume:least-squares-good} holds and $\|\beta_w\|_2 \asymp 1$, then $\sqrt{n}(\what{\theta} - \theta) = O_P(1)$.
\end{theorem}
\noindent The proof of this theorem is in Section~\ref{sec:linear-moments-proof}.

The key difference in the high dimensional regime where $d/n \to \kappa > 0$ is that $\trace{\var(\what\beta_1 - \what\beta_0)}/n \asymp 1$, because while $\var(\what{\beta}_w) \asymp (1/n)I_d$, the number of terms summed in the trace is growing at the same rate as $n$.

By further specifying the data generating model to assume, as in \eqref{eq:linear-model}, that the covariates are Gaussian and the potential outcomes are homoskedastic, we can get a more exact characterization of the variance, that precisely highlights the effect of being in the high dimensional regime. For example, Proposition~\ref{prop:gaussian-nuisance-exact} gives an exact expression for \smash{$\var(\what{\beta}_w\mid N_{1w} )$}. Plugging this in to the expression for the variance of $\what{\theta}$ in Theorem~\ref{thm:linear-moments} shows precisely how the variances diverge from one another under~\eqref{eq:linear-model} when $\Sigma=I$: for $n_2 / n \to p,~n_{1w} /n \to p_{1w}$,
\begin{equation*}
    \var(\sqrt{n}(\what\theta - \theta) \mid N_{11} = n_{11}, N_{10} = n_{10}) \to \frac{\|\beta_1 - \beta_0\|_2^2}{p} + \frac{\sigma^2}{1-\kappa} \|\E[X]\|_2^2(p_{11}^{-1} + p_{10}^{-1}) + \sigma^2 \kappa/(n_2(1-\kappa)).
\end{equation*}
As $\kappa \to 0$, the second term decreases and the third term disappears. Therefore, these differences correspond to excess variance not typically captured by large sample theory. Regardless of how realistic all of the normality assumptions are, they are helpful for precisely understanding the impact of high dimensional data on the statistical behavior of the estimator $\what\theta$.
\begin{proposition}
If the data follow the model in \eqref{eq:linear-model}, conditional on $N_{1w} = \sum_{i \in I_2} \ind{W_i=w}$, $\what{\beta_1}$ and $\what{\beta_0}$ are uncorrelated, and as long as $N_{1w} > d-1$, the variance of $\what{\beta}_w$ is
\begin{equation*}
    \var\left(\what{\beta}_w\mid N_{1w} \right) = \E\left[\left(\sum_{i=1}^n \ind{W_i=w} X_i X_i^\top\right)^{-1}\right] = \frac{\Sigma^{-1}}{N_{1w} - d - 1}.
\end{equation*}
\label{prop:gaussian-nuisance-exact}
\end{proposition}
The proof of Proposition~\ref{prop:gaussian-nuisance-exact} is deferred to Section~\ref{sec:proof-prop-gaussian-nuisance-exact}. The challenge of the proof is that $\sum_{i=1}^n \ind{W_i=w} X_i X_i^\top$ is nearly a Wishart distribution with $N_{1w}$ degrees of freedom, $\mathcal{W}(N_{1w}, \Sigma)$, except that $X_i \mid W_i=w$ is no longer Gaussian whenever $\gamma > 0$. The key to showing the result is to exploit the symmetry of $P(W_i \mid X_i=x)$ and the multivariate Gaussian distribution to show that the expectation of any function of $\sum_{i=1}^n \ind{W_i=w} X_i X_i^\top$ is equal to the expectation of that function of $W \sim \mathcal{W}(N_{1w}, \Sigma)$, which is shown in Lemma~\ref{lem:symmetry-wishart} of Section~\ref{sec:proof-prop-gaussian-nuisance-exact}.

\subsection{AIPW}
When the outcome is well-specified by a linear model, the bias in the outcome imputation model will be zero, even in high dimensions, and therefore, the AIPW estimator will be unbiased. Note that in the high dimensional regime, the estimator is no longer doubly robust in the sense of parametric specification: when the logistic regression model for the propensity score is correctly specified, but the ratio $d / n \to \kappa > 0$, the AIPW estimator $\what{\theta}_{\aipw{}}$ with MLE estimates of the nuisance parameters will only be unbiased if the outcome imputation models $\mu_w(x)$ are also correctly specified.

Additionally, the remainder term between $\what{\theta}_{\oracle{}}$ and $\what{\theta}_{\aipw{}}$ is no longer lower order when $d \asymp n$. To simplify exposition, assume that $\pi(x)$ is known, so that $\what{\pi}(x) = \pi(x)$. Then, the remainder term
\begin{equation*}
    \what{\theta}_{\oracle{}} - \what{\theta}_{\aipw{}} = \frac{1}{n} \sum_{i=1}^n \frac{\pi(X_i) - W_i}{\pi(X_i)} (\what{\beta}_1 - \beta_1)^\top X_i + \frac{\pi(X_i) - W_i}{1 - \pi(X_i)} (\what{\beta}_0 - \beta_0)^\top X_i
\end{equation*}
is no longer lower order, according to the following theorem:
\begin{theorem}
\label{thm:aipw-real-vs-oracle}
Let $\what{\theta}_{\aipw{}}$ be as defined in \eqref{eq:aipw} with $\what{\mu}_w(\cdot)$ defined according to \eqref{eq:linear-outcome-mmse}, and let $\what{\theta}_{\oracle{}}$ be as in \eqref{eq:aipw-oracle-linear}, and $\theta = \E[Y(1) - Y(0)]$. Assume the data follows the structural model in Eqs. \eqref{eq:prop-structural}-\eqref{eq:control-structural}, with $g^{-1}(t) = t$, $\var(X) = \Sigma$, having a bounded condition number.
Then, there exists a sequence of events $G_n$ measurable with respect to $(X_i, W_i)_{i \in \mathcal{I}_1}$, such that $P(G_n) \to 0$, $\E[\what{\theta} \mid G_n] = \theta$, and
\begin{gather*}
    \var\left( \sqrt{n}(\what{\theta}_{\aipw{}} - \what{\theta}_{\oracle{}}) \mid G_n \right) = \operatorname{trace}(\var(\what{\beta}_1\mid G_n) \Sigma_{Z^1}) + \operatorname{trace}(\var(\what{\beta}_0\mid G_n)  \Sigma_{Z^0}), \\
    \E\left[ (\what{\theta}_{\aipw{}} - \what{\theta}_{\oracle{}}) \what{\theta}_{\oracle{}} \mid G_n\right] = 0.
\end{gather*}
\end{theorem}
\noindent See Section~\ref{sec:aipw-real-vs-oracle-proof} for proof.

While each term of $\var(\what{\beta}_w)$ scales like $\frac{1}{n}$ (recall Theorem~\ref{thm:linear-moments} and Proposition~\ref{prop:gaussian-nuisance-exact}), the trace is the sum of $d$ such terms, so $\var(\sqrt{n}(\what{\theta}_{\oracle{}} - \what{\theta}_{\aipw{}})) \asymp \frac{d}{n}$. As claimed, this is lower order when $d = o(n),$ but in the present setting, this converges to a nonzero constant, showing that the realizable AIPW estimator $\what{\theta}_{\aipw{}}$ will not achieve the semiparametric efficiency bound derived by \citet{hahn98} for when $d$ is fixed.

It's instructive to compare the variances of the AIPW and G-computation estimators. As previously mentioned, the G-computation estimator will have lower variance than the AIPW in a parametric setting, even in the large sample asymptotic setting. Here, we will focus on comparing the additional variance from being in high dimensions. The additional variance terms (implicitly conditional on $G_n$) that don't vanish for G-computation and the AIPW estimator, respectively, are
\begin{align*}
    \trace{}(\var(\what{\beta}_1) \Sigma_X)+\trace{}(\var(\what{\beta}_0) \Sigma_X)~~\mbox{and}~~\trace{}(\var(\what{\beta}_1) \Sigma_{Z^1})+\trace{}(\var(\what{\beta}_0) \Sigma_{Z^0}).
\end{align*}
Focusing on $\trace{}(\var(\what{\beta}_1) \Sigma_X)$ versus $\trace{}(\var(\what{\beta}_1) \Sigma_{Z^1})$, note that $\var(Z^1) \succeq \var(X)$, and the two are equal only if $\pi(x) = 0.5$. Therefore, the variance gap from being in high dimensions will be higher for the AIPW estimator than G-computation in most observational data.

Finally, we note that in the above derivations, we assumed that the propensity score $\pi(x)$ is a known function. In observational data, this is rarely the case. Even when the propensity score is known to follow a generalized linear model, such as the logistic regression model or probit regression model, the high dimensional nature of the covariates creates variation in $\what{\pi}(x)$ that does not disappear asymptotically, analogous to the results of Proposition~\ref{prop:prediction-variance}. This creates additional statistical variation for the realizable AIPW estimator $\what{\theta}_{\aipw{}}$ at the $\sqrt{n}$ rate, leading to additional excess variance between G-computation and the AIPW estimator in this setting. While one could apply results for the estimation of the parameters of logistic regression in high dimensions \cite{SurCa2019}, and characterize this additional variance, it is messier than it is instructive. Therefore, we emphasize the empirical simulation results showing the variance gap between the AIPW estimator with a known versus estimated propensity score in Section~\ref{sec:linear-experiments} to demonstrate this phenomenon.

\subsection{TMLE}
The high dimensional setting that we consider here provides a different theoretical lens to study differences between the AIPW and TMLE estimators, as the variation in the nuisance parameter estimates creates behavior that is not captured in the traditional asymptotic analyses. In fact, as we will see, the AIPW and TMLE estimators will not generally be asymptotically equivalent in this setting. For example, while traditional assumptions and analyses suggest that the realizable AIPW and TMLE estimators both converge to the oracle AIPW estimator, we showed in Section~\ref{sec:linear-aipw} that the realizable AIPW estimator has higher asymptotic variance than oracle AIPW estimator. An exact characterization of the difference between the AIPW and TMLE estimators is challenging due to the nonlinear dependence of the estimator on the nuisance parameters, however we discuss the important differences in behavior via an illustrative first-order approximation.
 
By plugging in the closed form expression for $\what{\epsilon}_w$, the relationship between the TMLE and AIPW estimators is immediate. Indeed, because the estimates for $\what{\epsilon}_w$ each factorize as
\begin{equation*}
\what{\epsilon}_0 = \frac{\sum_{i \in \mathcal{I}_2} \frac{1-W_i}{1-\what{\pi}(X_i)}(Y_i - \what{\mu}_0(X_i))}{\sum_{i \in \mathcal{I}_2} \frac{1-W_i}{(1- \what{\pi}(X_i))^2}},
~~\mbox{and}~~
\what{\epsilon}_1 = \frac{\sum_{i \in \mathcal{I}_2} \frac{W_i}{\what{\pi}(X_i)}(Y_i - \what{\mu}_1(X_i))}{\sum_{i \in \mathcal{I}_2} \frac{W_i}{\what{\pi}^2(X_i)}},
\end{equation*}
plugging these into the expression for $\what{\tau}_{\tmle{}}$ gives
\begin{equation}
\begin{aligned}
    \what{\tau}_{\tmle{}} &= \frac{1}{n_2} \sum_{i\in \mathcal{I}_2} \what{\beta}_1^\top X_i + \frac{\sum_{i \in \mathcal{I}_2} \frac{1}{\what{\pi}(X_i)}}{\sum_{i \in \mathcal{I}_2} \frac{W_i}{\what{\pi}^2(X_i)}} \frac{W_i} {\what{\pi}(X_i)}(Y_i - \what{\beta}_{1}^\top X_i) \\
    &~~~~~~~~~~~~~ - \what{\beta}_0^\top X_i - \frac{\sum_{i \in \mathcal{I}_2} \frac{1}{1-\what{\pi}(X_i)}}{\sum_{i \in \mathcal{I}_2} \frac{1-W_i}{(1-\what{\pi}(X_i))^2}} \frac{1-W_i} {1-\what{\pi}(X_i)}(Y_i - \what{\beta}_{0}^\top X_i).
\end{aligned}
\end{equation}
While an intimidating expression altogether, it has a recognizable relationship to the AIPW estimator, with the residuals reweighted by a term like
\begin{equation*}
    T_n = \frac{\tfrac{1}{n_2}\sum_{i \in \mathcal{I}_2} \frac{1}{\what{\pi}(X_i)}}{\tfrac{1}{n_2}\sum_{i \in \mathcal{I}_2} \frac{W_i}{\what{\pi}^2(X_i)}}.
\end{equation*}
The two estimators will be asymptotically equivalent whenever $T_n \cp 1$, by Slutsky's Theorem. This is implied, for example, by overlap and the assumption that $\|\what{\pi}(\cdot) - \pi(\cdot)\|_{2,P} = o_P(n^{1/4})$, so that $\tfrac{1}{n_2}\sum_{i \in \mathcal{I}_2} \frac{W_i}{\what{\pi}^2(X_i)} \cp \E[1/\pi(X)]$ (and similarly for the numerator). However, in the current setting, $T_n \not \to 1$, because the asymptotic variation in $\what{\pi}$ means that while the numerator converges to $\E[1/\pi(X)]$, the denominator converges to something larger, so $T_n$ will typically converge to something less than 1.

Interestingly, this analysis would suggest that when the propensity score is estimated, the TMLE will be between the G-computation estimate and the AIPW estimate, but when the propensity score is known, it will be much more similar to the AIPW estimate. This intuition bares out in the simulation experiments in Section~\ref{sec:linear-experiments} below.

\subsection{Simulation Comparison of ATE Estimators}
\label{sec:linear-experiments}
The above calculations provide insights about the broad nature of different ATE estimators in this high dimensional setting, and how the difficulty with estimating the nuisance parameters creates additional asymptotic variance in the estimators. However, it is easier to interpret the magnitude of these differences and their relevance to practice by applying these estimators in simulation.

We simulate data from the model described in the beginning of this section (Section~\ref{sec:linear}) with the parameters and details described here. In this simulation, for each $n$, the coefficients $\beta_w$ of the potential outcomes and $\eta$ of the propensity model are proportional, with the pattern
\begin{equation*}
    \eta_i = \begin{cases}
    c & i \le \tfrac{n}{8},\\
    -c & \tfrac{n}{8} < i \le \tfrac{n}{4}, \\
    0 & \text{otherwise}.
    \end{cases}
\end{equation*}
The outcome noise variance is fixed at $\sigma^2 = 1$, $X$ follows a normal distribution $\normal{}(0, I_d)$, and $c$ is chosen separately for each model so that $\gamma^2 = 5$ and the $R^2$ of the treated and control outcome models are $0.9$ and $0.8$, respectively. With these choices, the simulated data has strong confounding from the baseline covariates, and predictions typical of the amount of signal that machine learning estimators can extract from real data science applications. The different scales on the parameters of the potential outcomes introduces a small amount of heterogeneity in the treatment effects, but the main treatment effect comes from adding a constant offset of $\tau$ to $Y(1)$ and no offset for $Y(0)$.

\begin{table}[ht]
    \centering
    \begin{tabular}{llrrrr}
\toprule
Propensity Est. & & G-estimator &    IPW &   AIPW &   TMLE \\
\midrule
            MLE &       Bias &          0.000 &  -0.020 &  0.001 &  0.003 \\
             &  Std. Err. &        0.040 &  0.252 &  0.069 &  0.068 \\
          Platt &       Bias &       0.001 &  0.021 &  0.003 &  0.003 \\
           &  Std. Err. &       0.041 &  0.186 &  0.064 &  0.064 \\
         Propensity Oracle &       Bias &       0.001 & -0.002 & -0.003 &     0.000 \\
          &  Std. Err. &        0.04 &  0.249 &  0.063 &  0.066 \\
         Oracle &       Bias &          -- &     -- & -0.000 &     -- \\
          &  Std. Err. &          -- &     -- &  0.061 &     -- \\
\bottomrule
\end{tabular}
    \caption{Comparing estimators for ATE in the linear model of Section~\ref{sec:linear}. The realizeable AIPW, TMLE, and oracle AIPW are all unbiased with similar standard errors; the G-computation has slightly smaller standard errors, and the IPW has much larger standard errors. In these simulations, $n = 4000$ and $p = 8$ so that $\kappa$ is practically small, representing the ``fixed dimension'' setting traditionally studied in asymptotic statistics. The behavior in this setting is similar to the asymptotic behavior predicted in many settings, including high dimensional settings with enough sparsity or structural assumptions to allow the nuisance parameter estimates to eventually converge to the true nuisance parameters.}
    \label{tab:linear-fixed-results}
\end{table}

\begin{table}[ht]
    \centering
\begin{tabular}{llrrrr}
\toprule
Propensity Est. & & G-estimator &    IPW &   AIPW &   TMLE \\
\midrule
            MLE &       Bias &      -0.001 & -3.018 & -0.028 & -0.001 \\
             &  Std. Err. &       0.045 &   2.930 &  0.657 &  0.126 \\
          Platt &       Bias &      -0.001 &  0.726 & -0.001 & -0.001 \\
           &  Std. Err. &       0.045 &  0.092 &  0.056 &  0.057 \\
           SLOE &       Bias &      -0.001 & -0.013 & -0.004 & -0.001 \\
            &  Std. Err. &       0.045 &  0.358 &  0.110 &  0.076 \\
         Propensity Oracle &       Bias &      -0.001 & -0.003 & -0.002 & -0.002 \\
          &  Std. Err. &       0.045 &  0.228 &  0.072 &  0.075 \\
         Oracle &       Bias &          -- &     -- & -0.002 &     -- \\
          &  Std. Err. &          -- &     -- &  0.065 &     -- \\
\bottomrule
\end{tabular}
    \caption{Comparing estimators for ATE in the linear model of Section~\ref{sec:linear}. Unlike Table~\ref{tab:linear-fixed-results}, the variance of each method is not similar, and depends strongly on the propensity score estimation method. In these simulations, $n = 4000$ and $p = 320$ so that $\kappa = 0.08$. Propensity score estimator is one of logistic regression maximum likelihood (MLE), post-hoc Platt scaling (Platt), or high-dimensionality correction as in \citet{YadlowskyYuMcDA21} (SLOE).}
    \label{tab:linear-hd-results}
\end{table}

We ran $1000$ simulations using different random seeds, and fit the four estimators discussed above---the G-computation, IPW, AIPW, and TMLE estimators to the simulated data. The nuisance parameters were estimated using $5$-fold cross fitting, as described in~\cite{ChernozhukovChDeDuHaNeRo16}, a more sophisticated form of sample splitting than discussed above that better preserves the sample size used for estimation. With 5 folds, each nuisance parameter estimate is constructed with (approximately) $80\%$ of the samples, so the effective ratio $\kappa$ is $25\%$ larger than the nominal value $d/n$. We also compared the estimators to the oracle AIPW estimator, which uses the true nuisance parameters instead of the estimated ones in Eq.~\eqref{eq:aipw}. Standard asymptotic arguments assume conditions so that the realizable AIPW estimator with estimated nuisance parameters is asmptotically equivalent to the oracle AIPW estimator, and show that the oracle AIPW estimator is semiparametric-efficient, meaning that no estimator can have lower variance than this estimator without making parametric assumptions, making a standard basis of comparison for realizable estimators.

To estimate the bias of each estimator, we compared the empirical average of the estimates on each simulation to the known ATE $\tau$. To estimate the standard error, we computed the standard deviation of the estimates across simulation runs. First, Table~\ref{tab:linear-fixed-results} shows results for the ``fixed dimension'' asymptotics, where $d = 8$ is constant as $n$ grows large (i.e. $\kappa=0)$. In this setting, we can see that the realizable AIPW estimator, TMLE, and oracle AIPW estimator all are unbiased, and have approximately the same variance. By taking advantage of the parametric nature of the data, parametric G-computation is unbiased and has lower variance than any of the other estimators, as expected. Then, in Table~\ref{tab:linear-hd-results}, we show the high dimensional asmyptotics studied above, with $n = 4000$ and $d = 320$, so that $\kappa = 0.08$. As with the fixed dimension results, G-computation estimator performs the best. However, the behavior of the realizable and oracle AIPW estimators diverge from one another, and from the TMLE. %

\section{Logistic outcome model}
\label{sec:logistic}
An interesting nature question is whether the results of the above setting generalize to any other models for the outcome, beyond the humble linear model, starting with other generalized linear models. Given the importance of binary outcomes in data science applications, we focus on logistic regression. \citepos{SurCa2019} recent breakthrough on statistical inference and asymptotics in the high dimensional regime that we study provide important results on which we rely. They provide results for inference on the parameters of the logistic regression model, which under appropriate structural and causal assumptions have a causal interpretation as the effect of the treatment on the  conditional odds ratio of the outcome. In contrast, we focus on the average treatment effect; in cases where we show that the asymptotic bias of the average treatment effect is zero, the delta method would allow extension of our results to the treatment effect on the marginal odds ratio.

Compared to the linear model results from Section~\ref{sec:linear}, we emphasize the challenges of the nuisance parameter estimation for logistic regression that lead to non-negligible bias if not handled correctly. In particular, we find that either (a) plugging in the MLE estimates of the nuisance parameters, or (b) re-calibrating the nuisance parameter estimates before performing ATE estimation lead to biased ATE estimates. These results emphasize statistical issues with a commonly used approach in practical applications--the re-calibration of the propensity score. However, through bias correction for the high dimensional data, we can still achieve approximately unbiased ATE estimates in the logistic regression model. 

\subsection{Nuisance parameter estimation}

Understanding the challenges of estimating the nuisance parameters, $\mu_w(\cdot) = \E[Y(w) \mid X=\cdot]$ and $\pi(\cdot)$ is important to understanding the challenges of estimating the ATE in this setting. When these follow a logistic regression model, the standard way to estimate the parameters $\beta_1$, $\beta_0$ and $\eta$ is via maximum likelihood estimation (MLE; we will abuse this abbreviation and use it for the maximum likelihood estimate, as well). In this section, we will discuss estimation of $\eta$ as a stand-in for all the parameters, $\beta_1$, $\beta_0$ and $\eta$.

\citet{SurCa2019} show that in this setting where $d / n \to \kappa$ as $n \to \infty$, when the MLE exists, the estimated parameters (e.g. $\what{\eta}$) converge to a Gaussian distribution with a biased center that overestimates the magnitude of the coefficients, $\what{\eta}_j - \alpha \eta_j \to \normal{}(0, \sigma_\star^2)$ for some $\alpha > 1$ and $\sigma_\star > 0$ that depends on $\kappa$ and the signal strength $\gamma^2 = \var(\eta^\top X)$. Therefore, for a test point $x$, the prediction $g(\what{\gamma}^\top x)$ is biased towards extreme predictions closer to $0$ or $1$ than is appropriate, leading to biased and poorly calibrated predictions. \citet{ZhaoSuCa20} extend these results to arbitrary covariance matrices, and give a conjecture extending the results to predict the bias of the intercept estimated via the MLE, which is important for treatment effect estimation.

\citet{YadlowskyYuMcDA21} describe a method SLOE to estimate (as $\what{\alpha}$) the extent of the bias ($\alpha$ in the previous paragraph). Therefore, we can construct an approximately unbiased estimate of $\eta$ as $\widehat{\eta}_{\sloe{}} = \frac{1}{\what{\alpha}}\what{\eta}$ using this method.

Even if we could exactly correct for the bias in the parameter estimates, e.g. $\E[\frac{1}{\alpha}\widehat{\eta}] = \eta$, some bias comes from the remaining variance in $\what{\eta}^\top x / \alpha$ and the non-linearity of the link function $g(\cdot)$, as $\E[h(\what{\eta}^\top x / \alpha)] \not= h(\E[\what{\eta}^\top x / \alpha])$. Because the distribution of $\what{\eta}$ is known \citep{SurCa2019}, one might expect that this bias can be corrected. Using that $\frac{1}{\alpha}\what{\eta}_j - \eta_j$ is approximately Gaussian, then we have that $\what{\eta}^\top x / \alpha \approx \eta^\top x + G_x$, with $G_x \sim \normal(0, \sigma_x^2)$. So, we seek a function $\widetilde{h}$ such that $\E[\widetilde{h}(\eta^\top x + G_x)] = h(\eta^\top x),$ for each $x$. \citet{Stefanski89} shows that if such a function exists, then $z \mapsto h(z)$ must be an entire function over the complex plane. Unfortunately, we know that the logit inverse link function is not entire (for example, the denominator of $h(z)$ is zero at $z = \pi i$), so no such $\widetilde{h}$ exists. \citet{Stefanski89} shows that approximately unbiased functions sometimes exist, even when $h$ is not entire; we follow their strategy using the approximation sequence
\begin{equation*}
    \widetilde{h}_k(z) = \sum_{m=0}^k \frac{(-\sigma_x^2)^m}{2^{m}m!}h^{(2m)}(z).
\end{equation*}
The form for $k=1$ is simple and relatively accurate,
\begin{equation*}
    \widetilde{h}_1(z) = h(z) + \frac{\sigma_x^2}{2}\frac{\exp(z)(\exp(z) - 1)}{(\exp(z) + 1)^3}.
\end{equation*}
The needed parameter $\sigma_x^2$ can be approximated for each $x$ using the inferential results in \citet{SurCa2019}, or \citet{ZhaoSuCa20} for when the covariates are not isotropic.

Perhaps surprisingly, we can correct the Jensen's gap bias between $\E[1/h(\what{\eta}^\top x / \alpha )]$ and $1/h(\eta^\top x)$. Indeed, writing out $1/h(z) = 1 + \exp(-z)$ shows that this function is entire, and applying Theorem~1 from \citet{Stefanski89} shows
\begin{equation*}
    \E\left[1 + \exp\left(-\what{\eta}^\top x / \alpha - \tfrac{\sigma_x^2}{2}\right) \right] = 1 + \exp(\eta^\top x),
\end{equation*}
when $(\what{\eta} - \alpha \eta)^\top x$ is normally distributed. Therefore, we can recover unbiased estimates of the ATE using IPW with this correction.

\subsection{Tradeoff between bias and calibration of nuisance parameters}
\label{sec:logistic-dilemma}
\begin{figure}[th]
     \centering
    \caption{True coefficients (orange) and estimated coefficients according to (a) MLE, (b) post-hoc Platt-scaled, and (c) post-hoc bias correction for high dimensions, according to a well-specified logistic regression simulation with $\kappa = 0.2$ and $\gamma^2 = 5$. The coefficient estimates from the MLE are systematically too large, whereas after post-hoc Platt scaling, they are systematically too small. Bias correction with SLOE removes the systematic bias in the coefficients.}
     \begin{subfigure}[b]{0.45\textwidth}
         \centering
         \includegraphics[width=\textwidth]{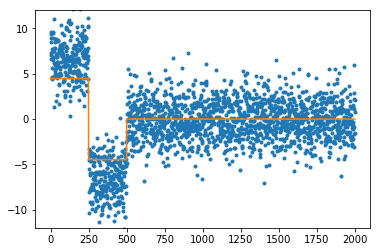}
         \caption{Unadjusted MLE estimate of parameters}
         \label{fig:mle-coefs}
     \end{subfigure}
     \hfill
     \begin{subfigure}[b]{0.45\textwidth}
         \centering
         \includegraphics[width=\textwidth]{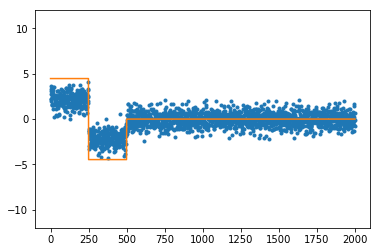}
         \caption{Post-hoc Platt scaling on 10\% of data held out}
         \label{fig:platt-coefs}
     \end{subfigure}
     \hfill
     \begin{subfigure}[b]{0.45\textwidth}
         \centering
         \includegraphics[width=\textwidth]{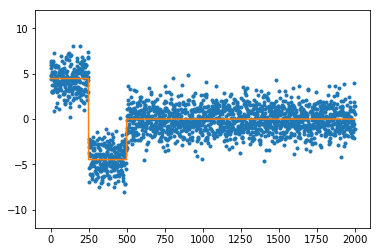}
         \caption{Bias correction for high-dimensions}
         \label{fig:unbiased-coefs}
     \end{subfigure}
    \label{fig:logistic-params}
\end{figure}

The unbiased prediction estimates described above will have statistical noise that does not disappear as $n \to \infty$. This noise effectively reduces the strength of the relationship between the predictions and outcomes, degrading the calibration of unbiased predictions on statistically independent test data. A common way to improve the calibration of learned predictions is to apply Platt scaling \cite{Platt98}; this method is especially sensible given that the true model is a logistic regression model. 

\citet{Platt98} defines the scaling in the context of using a support vector machine to estimate the coefficients, and then applying the scaling to convert them into coefficients in a logistic regression model to endow probabilistic interpretation. Here, we first estimated the coefficients with a logistic regression model, but the idea is still the same. To be concrete, we mean the following: given an estimate $\what{\eta}$ of the coefficients of the logistic regression model, and a statistically independent validation sample $\mathcal{I}_C$ of the input $X$ and output $W$ variables, $\{(X_i, W_i)\}_{i \in \mathcal{I}_C}$, compute a re-scaling of the coefficients $\what{t}$ to minimize the validation MLE,
\begin{equation*}
    \min_{t} \frac{1}{|\mathcal{I}_C|} \sum_{i \in \mathcal{I}_C} Y_i \log(h^{-1}(t \cdot \what{\eta}^\top X_i)) + (1-Y_i)\log(1 - h^{-1}(t \cdot \what{\eta}^\top X_i)).
\end{equation*}
Then, define the re-calibrated, ``Platt-scaled'' coefficients $\what{\eta}_{\text{cal}} = \what{t} \cdot \what{\eta}.$ In practice, a common approach is to estimate the coefficients using a leave-one-out cross-validation approach, where $\what{\eta}_{-i}$ is estimated from the full data, leaving out the $i$-th example, and then the recalibration scaler $\what{t}$ is estimated the solution to
\begin{equation*}
    \min_{t} \frac{1}{n} \sum_{i =1}^n Y_i \log(h^{-1}(t \cdot \what{\eta}_{-i}^\top X_i)) + (1-Y_i)\log(1 - h^{-1}(t \cdot \what{\eta}_{-i}^\top X_i)).
\end{equation*}
Using the results from \citet{SurCa2019} show that this will be consistent to the scaling estimated on a large, independent validation sample. Additionally, adapting techniques from \citet{YadlowskyYuMcDA21}, one can efficiently approximate the leave-one-out predictors and compute the re-scaling with little additional computation time or requirements.

Correcting the predictions to improve the calibration will reduce the error in the nuisance parameter estimates, so that for the recalibrated parameters $\what{\eta}_{\text{cal}}$, $\E[ (h(\eta^\top X) - h(\what{\eta}_{\text{cal}}^\top X))^2] < \E[ (h(\eta^\top X) - h(\widetilde{\eta}^\top X))^2]$. However, doing so re-introduces bias in the predictions, i.e. $\E[\what{\eta}_{\text{cal}}^\top x] \not= \eta^\top x$, which dominates the error of the downstream task of estimating the ATE. Therefore, recalibrating the predictions (either before or after applying the rescaling according to the ProbeFrontier method) will actually lead to worse ATE estimates.

\begin{figure}[ht]
    \centering
    \caption{Calibration of estimated predictions on independent and identically distributed data according to (a) MLE, (b) post-hoc Platt-scaled, and (c) post-hoc bias correction for high dimensions, according to a well-specified logistic regression simulation with $\kappa = 0.2$ and $\gamma^2 = 5$. The predictions from the MLE are significantly over-confident, while those from the bias-corrected coefficients are only mildly over-confident. Post-hoc Platt scaling properly calibrates the the \emph{predictions}, even though the \emph{coefficients} are systematically biased (see Fig.~\ref{fig:logistic-params}).}
     \begin{subfigure}[b]{0.45\textwidth}
         \centering
         \includegraphics[width=\textwidth]{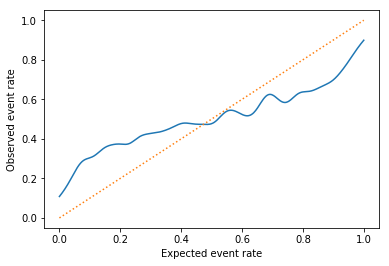}
         \caption{Unadjusted MLE estimate of parameters}
         \label{fig:mle-calib}
     \end{subfigure}
     \hfill
     \begin{subfigure}[b]{0.45\textwidth}
         \centering
         \includegraphics[width=\textwidth]{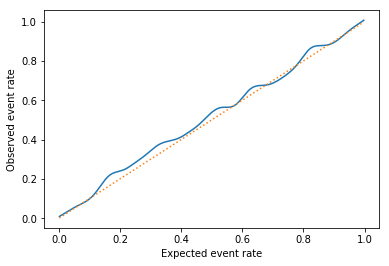}
         \caption{Post-hoc Platt scaling on 10\% of data held out}
         \label{fig:platt-calib}
     \end{subfigure}
     \hfill
     \begin{subfigure}[b]{0.45\textwidth}
         \centering
         \includegraphics[width=\textwidth]{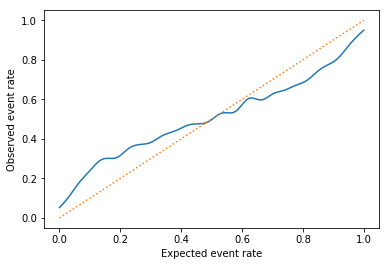}
         \caption{Bias correction for high-dimensions}
         \label{fig:unbiased-calib}
     \end{subfigure}
    \label{fig:logistic-calib}
\end{figure}

To visually understand these results, consider the results of the MLE, unbiased, and Platt-scaled (calibrated) estimates of the regression parameters from data simulated from a logistic regression model shown in Figure~\ref{fig:logistic-params}, and their calibration on an independent sample shown in Figure~\ref{fig:logistic-calib}. One can see that the MLE and Platt-scaled parameter estimates are biased, with the MLE biased towards overestimating the effect of the covariates, and the Platt-scaled estimates biased towards underestimating the effect of the covariates. However, neither the MLE or unbiased estimates lead to well-calibrated predictions, while the Platt-scaled estimates do. In this simulation, the signal strength $\gamma^2 = 5$ and dimension-to-example ratio $\kappa = 0.2$ are modest, and larger values of $\gamma^2$ or $\kappa$ would exacerbate these results even more.

\subsection{TMLE for binary outcomes}
We discussed the TMLE for continuous outcomes in Section~\ref{sec:tmle-linear} for the linear outcome regression model.
\citet[Chp. 4]{VanDerLaanRo11} compute a doubly robust estimate of the ATE for binary outcomes by estimating $\epsilon_1$ and $\epsilon_0$ via maximum likelihood regression estimates of the fluctuations $\frac{W}{\what{\pi}(X)}$ and $\frac{1-W}{1- \what{\pi}(X)}$, respectively, in a logistic regression model offset by the outcome regression. 

The nice property of TMLE with Gaussian family and identity link does not precisely translate when using the Binomial (logistic) link. However, as we see in the experiments, the link function does not make much of a difference in terms of the observed properties of the estimator.

\subsection{Experiments}
\label{sec:logistic-experiments}

The exact bias and variance of ATE estimators for the logistic regression outcome model are difficult to derive due to the nonlinearity. We saw in Section~\ref{sec:logistic-dilemma} that finding approximately unbiased estimates of the nuisance parameters is possible, and the intuition from Section~\ref{sec:linear} suggests that this will lead to good ATE estimates, with higher variance than anticipated from the asymptotic variance predicted by the standard semiparametric influence function. In this section, we show that this intuition bares out via simulations. We show that while the re-calibrated (Platt-scaled) nuisance parameter estimates are more accurate estimates of the true nuisance parameters, the unbiased estimates from rescaling as proposed by \citet{SurCa2019} lead to more accurate ATE estimates.

We used a very similar simulation setup as in Section~\ref{sec:linear-experiments}. The parameters $\beta_1,~\beta_0,$ and $\eta$ were scaled so that for each of them, $\gamma^2 = 5$. A constant offset is added to the log odds ratio to create a nontrivial average treatment effect. Note that this immediate induces heterogeneity in the treatment effect, because of the nonlinearity of the logistic inverse link function.

\begin{table}[h]
    \centering
    \begin{tabular}{llrrrrr}
\toprule
 & & & & &  \multicolumn{2}{c}{TMLE} \\
Logistic Estimator & & G-estimator &    IPW &   AIPW & Gaussian & Binomial \\
\midrule
            MLE &       Bias &      -0.013 & -0.014 & -0.004 &          -0.010 &          -0.008 \\
             &  Std. Err. &       0.012 &  0.042 &  0.029 &           0.015 &           0.025 \\
          Platt &       Bias &       0.144 &  0.079 &  0.052 &           0.052 &           0.050 \\
           &  Std. Err. &       0.012 &  0.012 &  0.012 &           0.012 &           0.012 \\
           SLOE &       Bias &      -0.004 & -0.000 & -0.000 &          -0.001 &          -0.001 \\
            &  Std. Err. &       0.012 &  0.019 &  0.018 &           0.015 &           0.015 \\
         Propensity Oracle &       Bias &      -0.004 &  0.001 &  0.000 &           0.000 &           0.001 \\
          &  Std. Err. &       0.012 &  0.020 &  0.014 &           0.014 &           0.014 \\
          Oracle &       Bias &         -- &    -- &  0.000 &             -- &             -- \\
          &  Std. Err. &         -- &    -- &  0.012 &             -- &             -- \\
\bottomrule
\end{tabular}
    \caption{Comparing estimators for ATE among the two choices for re-calibration of logistic regression models. In these simulations, $n = 8000$ and $p = 640$ so that $\kappa = 0.08$.}
    \label{tab:logistic-ate}
\end{table}

The results in Table~\ref{tab:logistic-ate} compare three estimators for the outcome models: the MLE, the MLE with post-hoc recalibration using Platt-scaling, as discussed in Section~\ref{sec:logistic-dilemma}, and the MLE rescaled using SLOE. For reference, we also include the estimates where the propensity score is known (with the outcome regression model estimated with SLOE) and the AIPW with oracle knowledge of the nuisance parameters.

\section{Discussion}
We studied, via theoretical investigation and simulation experiments, the variance of a variety of ATE estimators in a high dimensional regime where the nuisance parameters have non-vanishing variance, and therefore are inconsistent. We found that the effect of this is to inflate the variance, at the $\sqrt{n}$ scale, of the ATE estimates, so that none of the estimators achieve the semiparametric efficiency bound. We conjecture that because the estimated nuisance parameters will always have non-vanishing variance in this setting, that no estimator can achieve the semiparametric efficiency bound of \citet{hahn98}. Despite precisely characterizing the asymptotic variance of the estimators, we do not provide results for the asymptotic distribution of the estimators. The challenge here is that in high dimensions, asymptotic linearity of various parts of the estimators will not hold, preventing simple derivation of such results. Deriving such results is an interesting future direction, and may be possible with a leave-one-out analysis analogous to those in \citet{ElKarouiBeBiLiYu2013,LeiBiElKa18,SurCa2019}. 

Double robustness and large sample asymptotic linearity with the semiparametric efficient influence function does not provide guarantees about the variance of the estimator in this setting. Of the doubly-robust estimators that we tried, none have variance similar to the G-estimator. It's unclear, at this point, whether that is fundamental, or what the minimum variance doubly robust estimator would be.

We note that we have been imprecise with defining $\kappa$ with respect to our sample-splitting procedures. In all of our derivations, we assumed that the sample was split in two equal pieces and that $\kappa$ referred to the aspect ratio within each split. The variance of the nuisance parameter estimates, as well as the final ATE estimate, depend on $\kappa$ and therefore depend on the size of each split of the data. It's unclear in this high-dimensional regime whether cross-fitting, as discussed in \citet{ChernozhukovChDeDuHaNeRo16}, would affect the statistical properties of the estimator, as the high dimensionality could lead to nontrivial correlations between the splits. Beyond that, using more folds would reduce the effective $\kappa$ when fitting nuisance parameters, and therefore requires careful future consideration.

Finally, we end with noting that there are many connections between the doubly-robust methods described here and those that appear in off-policy policy evaluation in reinforcement learning (RL). Additionally, the way that predictions of future counterfactuals are plugged in to the parametric G-computation described here is similar to sequential G-computation methods and fitted Q-evaluation (FQE)  \citep{LeVoYu19}. In large scale empirical evaluations, \citet{VoloshinLeJiYu19} showed that FQE is more accurate at estimating a policy's value in high dimensional, long horizon regimes. This matches the intuition developed in this paper, and suggests that theory and lower variance methods developed for ATE estimation in the high dimensional regime may yield techniques applicable to RL, as well.

\section{Proofs}
\subsection{Proof of Proposition~\ref{prop:prediction-variance}}
\label{sec:prediction-variance-proof}
Standard results for linear regression show that so long as $N_{1w} \ge d$, \smash{$\var(\what{\beta}_1 \mid \{X_i, W_i\}) =  \sigma^2 \what{\Sigma}_w^{-1} / N_{1w}$}. Therefore, 
\begin{equation*}
    \var(x^\top \what{\beta}_1 \mid G) = \E\left[ \frac{1}{N_{1i}} x^\top \what{\Sigma}_1^{-1}
 x \mid G\right ] \ge \frac{\|x\|_2^2}{N_{1i}} \E[  \lambda_{\min}(\what{\Sigma}_1^{-1}) \mid G].
\end{equation*}

Because $d = \kappa n$, the assumption that $\|x\| \asymp \sqrt{d}$ implies that $\liminf_{n \to \infty} \|x\|_{2}^2 / N_{1i} > 0$, so it suffices to show that $\liminf_{n \to \infty} \E[  \lambda_{\min}(\what{\Sigma}_1^{-1}) \mid G] > 0$. Assuming the matrix inverse exists (it will almost surely on $G$), the properties of the inverse imply that $\lambda_{\min}(\what{\Sigma}_1^{-1}) = 1/\lambda_{\max}(\what{\Sigma}_1)$. Therefore, it suffices to show that $\liminf_{n\to\infty} P(0 < \lambda_{\max}(\what{\Sigma}_1)) > 0$.

For any fixed $N_1$, $\what{\Sigma}_1$ is the empirical covariance matrix of $N_1$ random vectors drawn from $X_i \mid W_i=1$, which are sub-Gaussian (proved below). These properties, along with a result from random matrix theory described in \citet{Vershynin12}, imply the following lemma, which bounds \smash{$\E[  \lambda_{\min}(\what{\Sigma}_1^{-1}) \mid G]$} from below.

\begin{lemma}
Under the assumptions of Proposition~\ref{prop:prediction-variance}, there exists finite $C, c > 0$ such that
\begin{equation*}
    \E[ \lambda_{\min}(\what{\Sigma}_1^{-1}) \mid N_1] \ge 
    \frac{1}{2\lambda_{\max}(\Sigma) + 2\left(\max\left\{\sqrt{\frac{\log(4)}{cN_1}} + C\sqrt{\frac{d}{N_1}}, 1\right\}\right)^2}
\end{equation*}
\label{lem:matrix-bound}
\end{lemma}

\noindent Proof of this lemma is deferred to Section~\ref{sec:proof-lem-matrix-bound}. Applying this result, and recalling that under $G$, $d < N_1 \le n$, we have that $\liminf_{n\to\infty} \E[\lambda_{\min}(\what{\Sigma}_1^{-1}) \mid G] \ge 1/(2(M +\kappa C^2))$.

\subsection{Proof of Lemma~\ref{lem:matrix-bound}}
\label{sec:proof-lem-matrix-bound}

Because $\what{\Sigma}_1$ is non-negative semi-definite,
$\lambda_{\min}(\what{\Sigma}_{1}^{-1}) = \frac{1}{\lambda_{\max}(\what{\Sigma}_1)}$.
Because this is a non-negative quantity, we can rewrite it's expectation as follows:
\begin{align*}
    \E\left[ \frac{1}{\lambda_{\max}(\what{\Sigma}_1)} \right] &= \int_0^\infty P\left(\frac{1}{\lambda_{\max}(\what{\Sigma}_1)} > t\right) \dif{t}
    = \int_0^\infty P\left(\frac{1}{t} > \lambda_{\max}(\what{\Sigma}_1) \right) \dif{t}
    \\
    &= \int_0^\infty P\left( \lambda_{\max}(\what{\Sigma}_1) < s\right) \frac{1}{s^2} \dif{s},
\intertext{
where the last line follows from the change of variables $s = 1/t$. Noting that $\|\what{\Sigma}_1-\Sigma_1\|_{\mathrm{op}} < s - \lambda_{\max}(\Sigma_1)$ implies $\lambda_{\max}(\what{\Sigma}_1) < s$, we can bound the above from below as}
    &\ge \int_{0}^\infty P\left( \|\what{\Sigma}_1-\Sigma_1\|_{\mathrm{op}} < s - \lambda_{\max}(\Sigma_1) \right) \frac{1}{s^2} \dif{s}
    \\
    &\ge \int_{\lambda_{\max}(\Sigma_1)+1}^\infty P\left( \|\what{\Sigma}_1-\Sigma_1\|_{\mathrm{op}} < s - \lambda_{\max}(\Sigma_1) \right) \frac{1}{s^2} \dif{s}.
    \intertext{Again, changing variables with $u^2 + \lambda_{\max(\Sigma_1)} = s$,}
    &=\int_{1}^\infty P\left( \|\what{\Sigma}_1-\Sigma_1\|_{\mathrm{op}} < u^2 \right) \frac{2u}{(u+\lambda_{\max(\Sigma_1)})^2} \dif{u}.
\end{align*}

Because $X_i$ is Gaussian, and $P(W_i) = 0.5$, the moment condition in Lemma 5.5 from \citet{Vershynin12} implies that conditional on $W_i=1$, $X_i$ are sub-Gaussian. Therefore, \citet[Theorem 5.39]{Vershynin12} implies that the eigenvalues of $\what{\Sigma}_1$ satisfy
\begin{equation*}
\|\what{\Sigma}_1 - \Sigma_1\|_{\mathrm{op}} \le \max\{\delta, \delta^2\}
\end{equation*}
with probability at least $1 - 2\exp(-ct^2)$, 
where $\delta = C \sqrt{\frac{d}{N_1}} + \frac{t}{\sqrt{N_1}}$. Because $u>1$, and for any $\delta\ge 1$, $\delta<\delta^2$, we have $\max\{\delta, \delta^2\} = u^2$ is only satisfied by $\delta=u$. Substituting the expression for $\delta$ and solving for $t$ gives
\begin{equation*}
    p(u) \defeq P\left( \|\what{\Sigma}_1-\Sigma_1\|_{\mathrm{op}} < u^2 \right) \ge 1 - 2\exp\left(-cN_1\left(u-C\sqrt{\frac{d}{N_1}}\right)\right)
\end{equation*}
When $u \ge \sqrt{\frac{\log(4)}{cN_1}} + C\sqrt{\frac{d}{N_1}}$, the above implies that $ p(u) \ge 1/2$. Therefore, truncating the limits of the integral again, we have
\begin{equation*}
    \int_{1}^\infty P\left( \|\what{\Sigma}_1-\Sigma_1\|_{\mathrm{op}} < u^2 \right) \frac{2u}{(u+\lambda_{\max(\Sigma_1)})^2} \dif{u} \ge \int_{\max\left\{\sqrt{\frac{\log(4)}{cN_1}} + C\sqrt{\frac{d}{N_1}},1\right\}}^\infty \frac{u}{(u+\lambda_{\max(\Sigma_1)})^2} \dif{u}.
\end{equation*}
This integral is straightforward to evaluate, giving the claimed result:
\begin{equation*}
    \int_{\max\left\{\sqrt{\frac{\log(4)}{cN_1}} + C\sqrt{\frac{d}{N_1}},1\right\}}^\infty \frac{u}{(u+\lambda_{\max(\Sigma_1)})^2} \dif{u} = \frac{1}{2\lambda_{\max}(\Sigma) + 2\left(\max\left\{\sqrt{\frac{\log(4)}{cN_1}} + C\sqrt{\frac{d}{N_1}},1\right\}\right)^2}.
\end{equation*}
Because $X$ is isotropic before conditioning on $W=1$, we know $\lambda_{\max}(\Sigma_1) \le \max\{2 \var( (\eta^\top X) h^{-1}(\eta^\top X)), 1\} < \infty$.

\subsection{Proof of Theorem~\ref{thm:linear-moments}}
\label{sec:linear-moments-proof}
For notational convenience, let $\overline{X} = \frac{1}{n_2}\sum_{i \in \mathcal{I}_2} X_i$.
Plugging \eqref{eq:treated-structural} and \eqref{eq:control-structural} into the equation for $\what\beta_w$ gives
\begin{align*}
    \E[\what{\beta}_w \mid G_n ] = \E[\E[\what{\beta}_w \mid \{X_i, W_i\}_{i=1}^n]] &= \E\left[\left( \what{\Sigma}_w \right)^{-1} \sum_{i=1}^n \ind{W_i=w} X_i \E[Y(w) \mid X=X_i] / N_{1w} \mid G_n \right]
    \\
    &= \E\left[\left(N_{1w} \what{\Sigma}_w \right)^{-1} \what{\Sigma}_w \beta_w \mid G_n \right] = \beta_w,
\end{align*}
so long as $\what{\Sigma}_w$ is invertible on $G_n$.
By the sample splitting procedure, $\what{\beta}_w$ and $G_n$ are independent of $\overline{X}$, so,
\begin{equation*}
    \E[\what\theta \mid G_n ] = \E[ \what{\beta}_1^\top \overline{X} - \what{\beta}_0^\top \overline{X}] = (\beta_1 - \beta_0)^\top \E[X] = \theta.
\end{equation*}
The variance follows from a simple application of the law of total variance, conditioning on $\overline{X}$, and some matrix manipulation:
\begin{align*}
    \var\left(\what\theta \right) &= \E\left[\var\left(\what\theta \mid \overline{X} \right)\right] + \var\left(\E\left[\what\theta \mid \overline{X} \right]\right)
    \\
    &= \E\left[\overline{X}^\top \var\left(\what\beta_1 - \what\beta_0 \right) \overline{X} \right] + \var\left((\beta_1 - \beta_0)^\top \overline{X} \right)
    \\
    &= \E\left[(\overline{X}^\top - \E[X]) \var\left(\what\beta_1 - \what\beta_0 \right) (\overline{X}^\top - \E[X]) \right] + \E[X] \var\left(\what\beta_1 - \what\beta_0 \right) \E[X] \\
    &~\quad~\quad + (\beta_1 - \beta_0)^\top \var\left(\overline{X} \right) (\beta_1 - \beta_0)
    \\
    &= \trace{}\left( \var\left(\what\beta_1 - \what\beta_0 \right)\Sigma \right) / n_2 + \E[X] \var\left(\what\beta_1 - \what\beta_0 \right) \E[X] \\
    &~\quad~\quad + (\beta_1 - \beta_0)^\top \Sigma (\beta_1 - \beta_0) / n_2,
\end{align*}
where the last equality follows from the linearity of the trace and expectation operators, and the cyclical property of the trace.

The remainder of the proof is to show that the additional conditions allow us to bound $\var( \what{\beta}_1 - \what{\beta}_0 \mid G_n )$ in high dimensions. If we can show that the maximum eigenvalue is bounded above, it implies that $\var(\what{\theta} \mid G_n) \le C/n$, for an appropriate constant $C<\infty$. Therefore, according to Chebyshev's inequality $\sqrt{n}(\what{\theta} - \theta) = O_P(1)$.

To do this, let $0 < \overline{p} < 1/2$ and
\begin{equation*}
    G_n = \{ \lambda_{\max}(\what{\Sigma}_1^{-1}) \le C_G,  \lambda_{\max}(\what{\Sigma}_0^{-1}) \le C_G, N_{1w} > \overline{p}n \},
\end{equation*}
for $C_G < \infty$ defined in Assumption~\ref{assume:least-squares-good}. This choice of $G_n$ ensures that $\what{\Sigma}_w$ is invertible, as needed above, and satisfies the conditions of the theorem. Standard results for linear regression with more observations than the dimension say
\begin{align*}
    \var\left( \what{\beta}_w  \mid \{X_i, W_i\}_{i=1}^n \right) &= \left( \sum_{i=1}^n \ind{W_i = w} X_iX_i^\top \right)^{-1} \left( \sum_{i=1}^n \ind{W_i = w} X_iX_i^\top \var(Y(w) \mid X=X_i) \right)\\
    &~\quad~\quad~\quad~\quad~\quad~\quad~\quad~\left( \sum_{i=1}^n \ind{W_i = w} X_iX_i^\top \right)^{-1}
    \\
    &\preceq \sigma_{\max}^2 \left( \sum_{i=1}^n \ind{W_i = w} X_iX_i^\top \right)^{-1}
\end{align*}
and that $\hat\beta_1$ and $\hat\beta_0$ are uncorrelated. Noting that on $G_n$, 
\begin{equation*}
    \left( \sum_{i=1}^n \ind{W_i = w} X_iX_i^\top \right)^{-1} \preceq \frac{C_G}{\overline{p}n} I_d,
\end{equation*}
we immediately can bound the variance as
\begin{equation*}
    \var\left( \what{\beta}_1 - \what{\beta}_0 \mid G_n\right) \preceq \frac{C_v}{ n} I_d,
\end{equation*}
where $C_v \defeq 2\sigma_{\max}^2 C_G/\overline{p} < \infty$.

\subsection{Proof of Proposition~\ref{prop:gaussian-nuisance-exact}}
\label{sec:proof-prop-gaussian-nuisance-exact}
The first equality of the proposition, that
\begin{equation*}
    \var\left( \what{\beta}_w \mid N_{1w}\right) = \sigma^2 \E\left[ \left( \sum_{i=1}^n \ind{W_i = w} X_iX_i^\top \right)^{-1} \mid N_{1w} \right],
\end{equation*}
is a standard result for homoskedastic linear regression with more observations than the dimension. Showing the second equality is more involved.

Due to the fact that the observations $i=1,\dots,n$ are all assumed to be i.i.d., we can assume, without loss of generality, that $W_i=w$ for observations $i=1, \dots, N_{1w}$ (otherwise, we could re-number the examples so that this holds). Then,
\begin{equation*}
    \E\left[ \left( \sum_{i=1}^n \ind{W_i = w} X_iX_i^\top \right)^{-1} \mid N_{1w} \right] = \E\left[\left( \sum_{i=1}^{N_{1w}} X_iX_i^\top \right)^{-1} \mid W_1=\dots=W_{N_{1w}}=w, N_{1w} \right].
\end{equation*}
From Bayes' rule, the conditional pdf of $X_i$ given $W_i = w$ is
\begin{equation*}
    p(X_i \mid W_i=w) = \frac{P(W_i = w \mid X_i)p(X_i)}{P(W_i=w)} = 2 h^{-1}(\eta^\top X_i) p(X_i),
\end{equation*}
and so
\begin{equation*}
    \E\left[\left( \sum_{i=1}^{N_{1w}} X_iX_i^\top \right)^{-1} \mid W_1=\dots=W_{N_{1w}}=w, N_{1w} \right] = \E\left[2^n \prod_{i=1}^{N_{1w}}h^{-1}(\eta^\top X_i) \left( \sum_{i=1}^{N_{1w}} X_iX_i^\top \right)^{-1} \mid N_{1w} \right].
\end{equation*}
Recall that because $h$ is the logistic link function, $h^{-1}(\eta^\top x) = 1-h^{-1}(-\eta^\top x)$. This symmetry allows us to derive the expectation on the RHS using standard results for the inverse Wishart distribution, according to the following lemma, whose proof is deferred to below.

\begin{lemma}
\label{lem:symmetry-wishart}
Let $f(X)$ be a function with the symmetry property $
    f(X) = 1 - f(-X),$ and let $(X_i)_{i=1}^n$ be a jointly Gaussian array of $d$-dimensional multivariate Gaussians, each distributed as $X_i \sim \normal{}(0, \Sigma)$.
Then, 
\begin{equation}
    \E\left[ \prod_{i=1}^n f(X_i) \left(\sum_{i=1}^n X_iX_i^\top\right)^{-1}\right] = \E\left[ 2^{-n} \left(\sum_{i=1}^n X_iX_i^\top \right)^{-1} \right].
    \label{eq:integrand-symmetry}
\end{equation}
\end{lemma}

This lemma implies that
\begin{equation*}
    \E\left[2^n \prod_{i=1}^{N_{1w}}h^{-1}(\eta^\top X_i) \left( \sum_{i=1}^{N_{1w}} X_iX_i^\top \right)^{-1} \mid N_{1w} \right] = \E\left[ \left( \sum_{i=1}^{N_{1w}} X_iX_i^\top \right)^{-1} \mid N_{1w} \right] = \E[ V^{-1} \mid N_{1w} ],
\end{equation*}
where $V\sim \mathcal{W}(\Sigma, N_{1w})$. The expectation of an inverse Wishart distribution is known,
\begin{equation*}
    \E[ V^{-1} \mid N_{1w} ] = \frac{\Sigma^{-1}}{N_{1w} - d -1},
\end{equation*}
which gives the main result,
\begin{equation*}
    \E\left[ \left( \sum_{i=1}^n \ind{W_i = w} X_iX_i^\top \right)^{-1} \mid N_{1w} \right] = \frac{\Sigma^{-1}}{N_{1w} - d -1}.
\end{equation*}

\begin{proof-of-lemma}[\ref{lem:symmetry-wishart}]
The probability distribution function (pdf) of a centered Gaussian satisfies $p_{\Sigma}(x) = p_{\Sigma}(-x)$. Therefore, for any (measurable) function $g : \R^d \to \R$,
\begin{equation*}
    \E[ g(X_i) ] = \int g(x) p(x) \dif{x} = \int g(x) \frac{(p_{\Sigma}(x) + p_{\Sigma}(-x))}{2} \dif{x}.
\end{equation*}
The key insight is that
\begin{align*}
    f(x_i)\frac{p_{\Sigma}(x_i)+p_{\Sigma}(-x_i)}{2} &= \frac{f(x_i) p_{\Sigma}(x_i)+(1-f(-x_i)) p_{\Sigma}(-x_i)}{2},
\end{align*}
and $(-x_i)(-x_i)^\top = x_ix_i^\top$, so applying the change of variables $x_i = -x_i$ in the following integral does not change the integrand:
\begin{align*}
    \E\left[ \prod_{i=1}^n f(X_i) \left(\sum_{i=1}^n X_iX_i^\top\right)^{-1}\right] &= \int \dots \int \prod_{i=1}^n f(x_i) \left(\sum_{i=1}^n x_ix_i^\top\right)^{-1} \prod_{i=1}^n p_{\Sigma}(x_i) \dif{x_1} \dots \dif{x_n} \\
    & = \int \dots \int \prod_{i=1}^n f(x_i)\frac{p_{\Sigma}(x_i)+p_{\Sigma}(-x_i)}{2} \left(\sum_{i=1}^n x_ix_i^\top\right)^{-1}  \dif{x_1} \dots \dif{x_n}
    \\
    & = \int \dots \int \prod_{i=1}^n \frac{f(x_i)p_{\Sigma}(x_i)+(1-f(x_i))p_{\Sigma}(x_i)}{2} \left(\sum_{i=1}^n x_ix_i^\top\right)^{-1}  \dif{x_1} \dots \dif{x_n}
    \\
    & = \int \dots \int \prod_{i=1}^n \frac{p_{\Sigma}(x_i)}{2} \left(\sum_{i=1}^n x_ix_i^\top\right)^{-1}  \dif{x_1} \dots \dif{x_n}.
\end{align*}
The last expression is exactly the integral $ = \E[ 2^{-n} (\sum_{i=1}^n X_iX_i^\top)^{-1}]$.
\end{proof-of-lemma}

\subsection{Proof of Theorem~\ref{thm:aipw-real-vs-oracle}}
\label{sec:aipw-real-vs-oracle-proof}
Let $G_n$ be defined as in Theorem~\ref{thm:linear-moments}, which ensures that the estimators $\what{\beta}_w$ are well-defined. All statements that follow in this proof will implicitly be conditional on $G_n$. By the law of total variance, the variance of the remainder term is
\begin{align*}
    \var\left(\sqrt{n}(\what{\theta}_{\oracle{}} - \what{\theta}_{\aipw{}}) \right)
    &=  \E\left[\var\left(\sqrt{n}(\what{\theta}_{\oracle{}} - \what{\theta}_{\aipw{}}) \mid \what{\beta} \right)\right]+ \\
    &~~~~~~ \var\left(\E\left[\sqrt{n}(\what{\theta}_{\oracle{}} - \what{\theta}_{\aipw{}}) \mid \what{\beta} \right]\right)
\end{align*}
Using $Z^{1}_i = \frac{\pi(X_i) - W_i}{\pi(X_i)} X_i$, $Z^{0}_i = \frac{\pi(X_i) - W_i}{1-\pi(X_i)} X_i$, and $\Sigma_{Z^w} = \var(Z^w)$, the conditional variance is simply 
\begin{align*}
 \var\left(\sqrt{n}(\what{\theta}_{\oracle{}} - \what{\theta}_{\aipw{}}) \mid \what{\beta} \right) &= 
    \var\left( \frac{1}{\sqrt{n}} \sum_{i=1}^n (\what{\beta}_1 - \beta_1)^\top Z^1_i + (\what{\beta}_0 - \beta_0)^\top Z^0_i \mid \what{\beta}  \right) \\
    &=
    \left(\what{\beta}_1 - \beta_1\right)^\top \Sigma_{Z^1} \left(\what{\beta}_1 - \beta_1\right) + \left(\what{\beta}_0 - \beta_0\right)^\top \Sigma_{Z^0} \left(\what{\beta}_0 - \beta_0\right).
\end{align*}
The conditional expectation $\E\left[\sqrt{n}(\what{\theta}_{\oracle{}} - \what{\theta}_{\aipw{}}) \mid \what{\beta} \right] = 0$, by the tower law and the fact that $\E[\pi(x) - W \mid X=x] = 0$. Putting these pieces together,
\begin{equation}
    \var(\sqrt{n}(\what{\theta}_{\oracle{}} - \what{\theta}_{\aipw{}})) = \operatorname{trace}(\var(\what{\beta}_1) \Sigma_{Z^1}) + \operatorname{trace}(\var(\what{\beta}_0)  \Sigma_{Z^0}).
\end{equation}
The term $\what{\theta}_{\aipw{}} - \what{\theta}_{\oracle{}}$ is uncorrelated with $\what{\theta}_{\oracle{}}$ due to the sample splitting procedure, which ensures that the $\what{\beta}_w - \beta_w$ are independent of the data averaged to estimate $\theta$.

\subsection*{Acknowledgements}
Many thanks to Alexander D'Amour, Victor Veitch, and Avi Feller, whose suggestions and conversations with me have greatly improved the paper.

\bibliographystyle{abbrvnat}
\bibliography{bib}

\end{document}